\documentclass[11pt,onecolumn]{article}
\usepackage{amsmath,amsthm,amssymb,amsfonts,latexsym,epsfig,graphicx}
\usepackage{algorithm,algorithmicx}
\usepackage{color}
\usepackage{dsfont}
\usepackage{wrapfig}
\usepackage{setspace}

\usepackage[nohead,textwidth=4.5in,textheight=7.125in,footskip=0.3in]{geometry}

\newtheorem{theorem}{Theorem}

\newtheorem{lemma}[theorem]{Lemma}
\newtheorem{definition}[theorem]{Definition}

\newtheorem{remark}{Remark}

\usepackage{graphicx}
\usepackage{amsmath}
\usepackage{amssymb}
\usepackage{IEEEtrantools}
\usepackage{enumerate}
\usepackage{url}

\newtheorem{prop}[theorem]{Proposition}

\def\calS{\mathcal{S}}
\def\calT{\mathcal{T}}

\numberwithin{equation}{section}

\begin{document}

\title{Superconcentrators of Density 25.3}

\author{
     Vladimir Kolmogorov$^{\tiny \dag}$ \\ {\normalsize\tt vnk@ist.ac.at}
\and Michal Rol\'inek$^{\tiny \dag}$ \\ {\normalsize\tt michal.rolinek@ist.ac.at}
\and \\ \\
\normalsize ~$^{\tiny \dag}$Institute of Science and Technology Austria \qquad
}
\date{}

\maketitle

\singlespacing

\begin{abstract} An $N$-superconcentrator is a directed, acyclic graph with $N$ input nodes and $N$ output nodes such that every subset of the inputs and every subset of the outputs of same cardinality can be connected by node-disjoint paths. It is known that linear-size and bounded-degree superconcentrators exist. We prove the existence of such superconcentrators with asymptotic density $25.3$ (where the density is the number of edges divided by $N$). The previously best known densities were $28$ \cite{Scho2006} and $27.4136$ \cite{YuanK12}. 
\end{abstract}

\section{Introduction}

\begin{definition} An $N$-\textit{superconcentrator} is a directed acyclic graph having exactly $N$ input nodes $I$ and $N$ output nodes $O$ with the following property: for every subset $S \subset I$ and every subset $T \subset O$ with $|S| = |T| = k$ there exist $k$ node-disjoint paths connecting the nodes in $S$ to the nodes in $T$ (in an arbitrary order).

The \textit{density} of an $N$-superconcentrator is the number of its edges divided by $N$.
\end{definition}

Superconcentrators of bounded degree and linear size have been known to exist \cite{Valiant1975, Pinsker73}. 
Their applications include lower bounds  of resolution proofs \cite{Urquhart1987,Schoning1997} and
constructions of graphs that are hard to pebble~\cite{paul1976,Lengauer82,Dziembowski:PoS}, which are
used e.g.\ in cryptographic protocols~\cite{DworkNW05,Dziembowski:PoS}.
In these applications it is important to have superconcentrators of smallest possible density. The best bounds for asymptotic densities have improved several times \cite{Pippenger77, Chung79, Bas81, Schoning2000} and now to our knowledge the best known bounds are $28$ \cite{Scho2006} and $27.4136$ \cite{YuanK12}. The smallest known density of an explicitly constructable superconcentrator is 44 \cite{AlonC03}. In this paper we show that $N$-superconcentrators of asymptotic density 25.3 exist. The best known lower bound for the asymptotic density is 5~\cite{LevValiant83}.

\vspace{3pt}
\noindent{\bf Overview of our techniques.~~}
We follow the construction of an $N$-superconcentrator $\Gamma_N$  introduced by Alon and Capalbo \cite{AlonC03}.
Its main building block is a bipartite graph $E_N$ with
certain properties. In \cite{AlonC03} this graph was required to be an {\em expander graph} with particular constants:

\begin{definition} Let $E_N$ be a bipartite graph with $N$ left vertices $L$ and $N$ right vertices $R$ and directed edges going from $L$ to $R$.
It is called an $(N,\alpha, \beta)$-\textit{expander graph} (where $\alpha, \beta \in [0,1]$) if for all subsets $S \subset L$ with $|S| = \lfloor \alpha N \rfloor$ it holds that:
$$|\Gamma(S)| \geq \lceil \beta N \rceil.$$
Here $\Gamma(S) \subset R$ is the set of neighbours of the nodes in $S$. 
\end{definition}

Sch{\"o}ning~\cite{Scho2006} showed that a random bipartite graph of degree $d=6$
satisfies the property in \cite{AlonC03} with high probability, thus proving the existence of a superconcentrator of asymptotic density $28$.

To get a smaller density, we show that the required expansion property of $E_N$ can be relaxed if the graph
satisfies an additional condition that we call a {\em pair expansion}. 
To describe the new condition, we assume that $N$ is even and the right vertices $R$ are grouped into pairs.  We say that a left vertex is \textit{adjacent to a pair} in $R$ if it is adjacent to at least one vertex in the pair. Similarly, a subset of left vertices $U \subset L$ is adjacent to a pair in $R$ if some $l \in U$ is adjacent to it.

\begin{definition} A directed bipartite graph with $L$ and $R$ as above and with vertices in $R$ grouped into pairs is a $(N,\alpha, \gamma)$-\textit{pair-expander} graph if for each $U \subset L$ with $|U| = k = \lfloor \alpha N \rfloor$ is adjacent to at least $\lfloor \gamma k \rfloor$ pairs.
\end{definition}

In the second part of the paper we prove that the new conditions
are satisfied with a high probability by a random bipartite graph of average degree $d=5.325$.
We follow the probabilistic argument of Bassalygo \cite{Bas81},
except that we use a fractional degree which presents an additional technical challenge.

Note that the argument in \cite{Bas81} uses an upper bound on the probability that a given subset $U\subset L$ does not satisfy
the expansion property.
As a side result, in Appendix~A we give an exact expression for this probability
as a sum with $O(N)$ terms. Our computational experiments, however, indicate that 
the bound is very close to the true value, and so we do not use this result
in our analysis.

\section{Construction}

We start by reviewing the construction of an $N$-superconcentrator $\Gamma_N$  of \cite{AlonC03} and \cite{Schoning2000}.
Graph $\Gamma_N$ for a sufficiently large $N$ is defined recursively as follows. 
Let $X$ and $Y$ be disjoint sets of $N$ vertices each. The input and output sets of $\Gamma_N$ are $X$ and $Y$, respectively. Let also $X' = \{x'_1,\dots, x'_N\}$ and $Y' = \{y'_1,\dots, y'_N\}$ be disjoint sets.

A copy of the graph $E_N$ discussed in the previous section is inserted between $X$ and $X'$. The resulting set of edges is called $\Lambda_X$;
these edges are directed from $X$ to $X'$.
Similarly, a copy of the {\em reverse} of graph $E_N$ is inserted between $Y'$ and $Y$, and
the resulting set of edges (directed from $Y'$ to $Y$) is called $\Lambda_Y$.

In addition, for each $i \in \{1,\dots, N/2\}$, the edges $(x'_{i+N/2}, y'_i)$, $(x'_{i+N/2}, x'_i)$, $(x'_i, y'_{i+N/2})$, and $(y'_i, y'_{i+N/2})$ are all in $\Gamma_N$.

Further let $X'' = \{x'_i \in X' | i \in \{1,\dots, N/2\}\}$ and $Y'' = \{y'_i \in Y' | i \in \{1,\dots, N/2\}\}$ and as edges between $X''$ and $Y''$ take edges of the superconcentrator $\Gamma_{N/2}$.

\begin{figure}
\begin{center}
\includegraphics{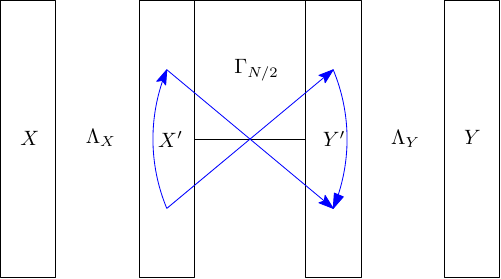}
\end{center}
\caption{Construction of superconcentrator $\Gamma_N$, figure adapted from \cite{Scho2006}.}
\label{fig:superconstr}
\end{figure}

This completes the description of graphs $\Gamma_N$. A schematic illustration is given in Fig.~\ref{fig:superconstr}. By construction, the number of edges $f(N)$ satisfies
$$f(N)  = (2d+2)N + f(N/2),$$
where $d$ is the average degree of $E_N$. Solving this recursion gives $$f(N)=4(d+1)N+const.$$

\begin{remark} 
Below we will work with piecewise linear functions. It will be convenient to specify them by a list of points:
the list $(x_1,y_1),\ldots,(x_k,y_k)$ with $x_1<\ldots<x_k$ specifies a continuous function $F:[x_1,x_k]\rightarrow\mathbb R$ which is linear on each interval $[x_i,x_{i+1}]$,
and satisfies $F(x_i)=y_i$ for all $i$.

\end{remark}

\begin{theorem}[\cite{AlonC03}]
Let $e(\alpha) \colon [0,1] \to [0,1]$ be a piecewise linear function (see also Fig. \ref{fig:ealplha}) connecting the points
$$(0,0), \qquad \left(\frac14,\frac12\right),\qquad \left(\frac12,\frac34\right), \qquad (1,1).$$
Suppose that $E_N$ is an $(N,\alpha,e(\alpha))$-expander for any $\alpha\in[0,1]$ and for any $N$.
Then $\Gamma_N$ is an $N$-superconcentrator.
\label{th:AlonC03}
\end{theorem}

\begin{figure}
\begin{center}
\includegraphics{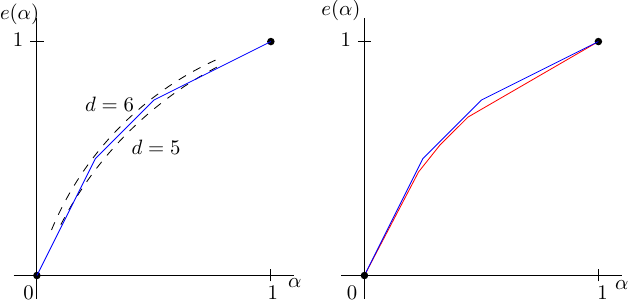}
\end{center}
\caption{Left: 
Expansion factor function $e(\alpha)$ requred by~\cite{AlonC03}. It is achieved by random bipartite graphs of average degree $d=6$, but not of degree $d=5$ (formula
for curves with $d=5$ and $d=6$ comes from \cite{Bas81} and is generalized in Section \ref{sec:expan} for fractional $d$).
Right: Comparison of $e(\alpha)$ from \cite{AlonC03} (blue) and $e(\alpha)$ we introduce (red).}
\label{fig:ealplha}
\end{figure}

As shown by \cite{Scho2006}, there exist graphs $E_N$ of degree $d=6$ that satisfy conditions of Theorem~\ref{th:AlonC03};
this yields superconcentrators of degree $4(6+1)+o(1)=28+o(1)$.

The vital part of verifying that $\Gamma_N$ is a superconcentrator boils down to constructing certain matchings (see Section \ref{sec:match}) from $\Lambda_X$ and $\Lambda_Y$ for given $S \subset X$ and $T \subset Y$ with $|S| = |T| = \alpha |N|$. This construction works in three regimes based on which subinterval of $[0,1]$ $\alpha$ falls into. Roughly speaking the three regimes correspond to how effectively can the overlaps of neighborhoods of $S$ and $T$ (when $X'$ and $Y'$ are identified) be used.

We require more from the first regime, namely also good pair-expansion. This can be used to construct some fraction of the sought matching cheaply. Even though this fraction decreases with $\alpha$, it allows to ``push down'' the curve of $e(\alpha)$ in the critical regions and thus we obtain a milder requirement on the degree of the random bipartite graph.

We also subdivide this interval corresponding to the first regime to two subintervals $[0,C_1]$ and $[C_1, C_3]$. This does not play a fundamental role, it only serves to obtain slightly better constants in the end.

Our alternative condition on $E_N$ is the following.

\begin{theorem}\label{mainthm} 
Let $C_1$, $C_2$, $C_3$, $C_4$, $C_5$, $C_6$ be real numbers from $(0,1)$ satisfying the following inequalities:
\begin{align}
C_1 < C_3 &< C_5 \label{numeq1}\\
C_2 + C_4 &\geq 1 \label{numeq2}\\
C_1+C_2+C_3 &\leq 1 \label{numeq3}\\
\frac{C_2}{C_1}> \frac{C_4-C_2}{C_3-C_1} > \frac{C_6-C_4}{C_5-C_3} = 1 &> \frac{1-C_6}{1-C_5} \label{numeq4}
\end{align}
Let $e(\alpha)$ be a piecewise linear function connecting the points
$$(0,0), \qquad (C_1, C_2), \qquad (C_3,C_4), \qquad (C_5,C_6), \qquad (1,1).$$
Suppose that for every $N$ graph $E_N$ is a bipartite graph
 with $N$ left vertices $\{x_1, \dots, x_N\}$ and $N$ right vertices $\{y_1, \dots, y_N\}$ and edges directed from left to right with the following properties:
\begin{enumerate}[(a)]
\item $E_N$ is a $(N, \alpha, e(\alpha))$-expander for every $\alpha\in[0,1]$.
\item $E_N$ is a $(N,\alpha ,1)$-pair-expander for every $\alpha \in [0,C_3]$ where the pairs are $(y_i, y_{i+N/2})$ for $i \in \{1,\dots, N/2\}$.
\end{enumerate}
Then $\Gamma_N$ is an $N$-superconcentrator.
\end{theorem}


Note that the (degenerate) choice of $C_1 = C_3 = \frac14$, $C_2 = C_4 = C_5 = \frac12$, and $C_6 = \frac34$ gives function $e(\alpha)$ from Theorem \ref{th:AlonC03}.

We are not able to give a direct combinatorial interpretation of conditions \eqref{numeq1}-\eqref{numeq4}, however one may spot that \eqref{numeq2} and \eqref{numeq3} enforce high enough expansion and \eqref{numeq4} witnesses for the concavity of $e(\alpha)$ which later translates into certain monotonicity of overlap sizes.

Following numerical experiments, we chose the values of constants $C_1$, \dots $C_6$ that effectively minimize the average degree $d$ within the bounds given by Theorem \ref{mainthm}.

\begin{theorem}\label{expexist} If $N$ is sufficiently large then there exists a bipartite graph $E_N$ of average degree $d=5.325$
that satisfies conditions of Theorem~\ref{mainthm} with constants
\begin{align*}
C_1 &= 0.2301, \quad C_3 = 0.3322, \quad C_2 = C_5 = 1-C_1-C_3, \\
C_4 &= 1-C_2, \quad C_6 = 1-C_3.
\end{align*}
\end{theorem}
Taken together, these theorems imply our main result, i.e.\ the existence of superconcentrators of density
$4(5.325+1)+o(1)=25.3+o(1)$.

\begin{figure}[t]
\begin{center}
\includegraphics[scale=1]{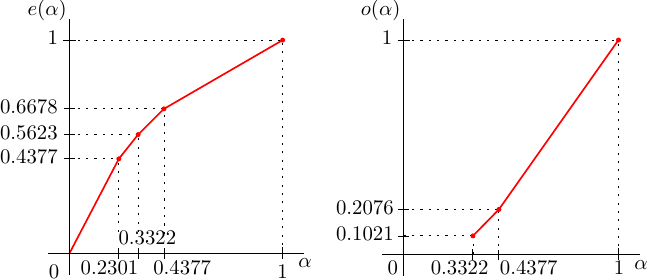} 
\end{center}
\caption{Left: function $e(\alpha)$ for the constants in Theorem~\ref{expexist}.
Right: function $o(\alpha)$ for these constants (it is used in the proof of Theorem~\ref{mainthm}).}
\label{fig:graphs}
\end{figure}

\begin{remark}
Note that an alternative construction was given in \cite{YuanK12}.
They modify the construction above by slightly shrinking the
size of sets $X'$, $Y'$, $X''$, $Y''$ while maintaining  $|X''|=|Y''|=\frac{1}{2}|X'|=\frac{1}{2}|Y'|$.
They also add extra edges from $X$ to $Y$ of the form $(x_i,y_i)$
for a small fraction of indices $i\in\{1,\ldots,N\}$. As a result, they obtain
superconcentrators of density $27.4136+o(1)$.

The analysis in \cite{YuanK12} uses only the ordinary expansion property, as in \cite{AlonC03}.
We conjecture that the pair-expansion property could also improve the density of the scheme in \cite{YuanK12},
but haven't explored the constants for such approach.
\end{remark}


\section{Proof of Theorem \ref{mainthm}} \label{sec:match}

Let us fix some $S \subset X$ and $T \subset Y$ such that $|S| = |T|$.

The following sufficient condition for $\Gamma_N$ to be a $N$-super\-con\-cent\-rator was established (and is easy to prove) in \cite{AlonC03}.

\begin{definition} We say that a matching $\Lambda$ between sets of vertices $A$ and $B$ \emph{saturates} some $A' \subset A$ if each vertex $x \in A'$ appears in some edge of $\Lambda$.
\end{definition}

\begin{lemma}\label{suffcond} 
Suppose that for any $S\subseteq X$ and $T\subseteq Y$ with $|S| =|T|$
there exist matchings $M^*_S \subset \Lambda_X$ and $M^*_T \subset \Lambda_Y$ such that both $M^*_S$ and $M^*_T$ have $|S| = |T|$ edges, and $M^*_S$ and $M^*_T$ satisfy the conditions stated below.

\begin{enumerate}[(a)]
\item $M^*_S$ saturates $S$ and $M^*_T$ saturates $T$.
\item Let $i \in \{1,\dots, N/2 \}$. Then if $M^*_S$ covers both $x'_i$ and $x'_{i+N/2}$, then $M^*_T$ covers at least one vertex of $\{y_i, y_{i+N/2} \}$. Similarly, if $M^*_T$ covers both $y'_i$ and $y'_{i+N/2}$, then $M^*_S$ covers at least one vertex of $\{x_i, x_{i+N/2} \}$.
\end{enumerate}

Then $\Gamma_N$ is a $N$-superconcentrator.
\end{lemma}

With Lemma \ref{suffcond} in mind, we devote the rest of this section to proving the following proposition.

\begin{prop} For any $S\subseteq X$ and $T\subseteq Y$ with $|S| =|T|$
there exist matchings $M^*_S$ and $M^*_T$ satisfying the conditions specified in Lemma \ref{suffcond}.
\end{prop}
\begin{proof}

Let us denote by $X'_S$ the neighborhood of $S$ in $X'$ and similarly $Y'_T$ the neighborhood of $T$ in $Y'$.

As in \cite{AlonC03}, we will construct the desired matchings from two auxiliary pairs of matchings. The first one exploits the overlap in indices between $X'_S$ and $Y'_T$.

Define function $o(\alpha)\colon [C_3,1] \to [0,1] $ (which will control the size of the overlaps) as a piecewise linear function connecting the points
$$(C_3,C_3-C_1), \qquad (C_5, C_5-C_1), \qquad (1,1).$$


\begin{lemma}\label{overlapmatch} Let $S$ and $T$ be as above. Then there exist matchings $M^1_S \subset \Lambda_X$ and $M^1_T \subset \Lambda_Y$, and a subset $I$ of $\{1,\dots, N\}$ that satisfy the following conditions.

\begin{enumerate}[(a)]
\item Each edge in $M^1_S$ is incident to a vertex in $S$ and each edge in $M^1_T$ is incident to a vertex in $T$.
\item Let $X'_I$ denote the subset of $X'$ of the form $\{x'_i | i \in I\}$, and similarly let $Y'_I = \{y'_i | i \in I\}$. Then $M^1_S$ saturates $X'_I$ and $M^1_T$ saturates $Y'_I$.

\item Let $\alpha = |S|/N = |T|/N$. If $\alpha \geq C_3$, then $|I| \geq o(\alpha)N$.

\end{enumerate}
\end{lemma}
\begin{proof} It suffices to prove the lemma only in the case $\alpha\ge C_3$.
(When $\alpha<C_3$, we can take $I=\emptyset$, then we only need to verify property (a).
Matchings satisfying this property can be obtained, for example, by applying the lemma to subsets $S'=X$, $T'=Y$ and taking the matchings induced by $S,T$.)

Replace the edges between $X'$ and $Y'$ by the edges $$\{(x'_i, y'_i), i \in \{1,\dots, N\}\}$$ and call the resulting graph $\Gamma^1_N$. Applying Menger's Theorem for $\Gamma^1_N$ gives:

\begin{center} The maximum possible number of vertex-disjoint paths from $S$ to $T$ is equal to the minimum possible cardinality of a set of vertices $C$ that separates $S$ and $T$ in $\Gamma^1_N$.
\end{center}

Note that the maximum possible number of vertex-disjoint paths from $S$ to $T$ equals the maximum size of the set $I$.
Now consider the minimum vertex cut $C$ and let $|C \cap S| = aN$, $|C \cap X'_S| = bN$, $|C \cap Y'_T| = cN$, and $|C \cap T| = dN$ for some $a$, $b$, $c$, $d$. It suffices to prove $a+b+c+d \geq o(\alpha)$.

If $a+d > o(\alpha)$, we are done. Otherwise, assume $a+d \leq o(\alpha)$ and by computing the sizes of the neighbourhoods of $S\setminus C$ and $T \setminus C$ we find that
$$b + c \geq e(\alpha-a) + e(\alpha - d) - 1$$
or otherwise some vertex in $S\setminus C$ could be connected to a vertex in $T \setminus C$.
From there we have
\begin{equation}\label{eq_overlap}
a+b+c+d \geq e(\alpha-a)+a + e(\alpha - d)+d - 1 \qquad \text{for} \quad a+d \leq o(\alpha).
\end{equation}

The condition (\ref{numeq4}) implies the slope of $e(\alpha)$ decreases at points $C_1$, $C_3$, and $C_5$ and that this slope is equal to one on $[C_3,C_5]$. From here it follows that for $\alpha > C_3$, the right-hand side of (\ref{eq_overlap}) attains its minimal value for $a=0$, $d=o(\alpha)$.
Therefore
$$a+b+c+d \geq e(\alpha) + e(\alpha - o(\alpha))+o(\alpha)-1.$$

Now we distinguish two cases.

\begin{itemize}
\item $C_3 \leq \alpha \leq C_5$: For these values of $\alpha$ we have $\alpha - o(\alpha) = C_1$ and since $e(\alpha)$ is increasing, the inequality
$$e(\alpha) + e(C_1)+o(\alpha)-1 \geq o(\alpha)$$
only needs to be verified for $\alpha = C_3$, where it reduces to (\ref{numeq2}).

\item $C_5 \leq \alpha \leq 1$: This time $\alpha - o(\alpha) \in [0,C_1]$ and the inequality
$$
e(\alpha) + e(\alpha-o(\alpha)) +o(\alpha)-1 \geq o(\alpha)
$$
is linear in $\alpha$. Verifying for $\alpha = 1$ is immediate and for $\alpha =C_5$ it was already handled in the first distinguished case.
\end{itemize}
\end{proof}

The second pair of matchings takes place in $\Gamma_N$ after merging some pairs of vertices so that the ``bad case'' from Lemma \ref{suffcond}(b) is avoided.

Let us merge the pairs of vertices $(x'_i, x'_{i+N/2})$ and $(y'_i, y'_{i+N/2})$ for those $i$ for which $i \notin I$ and $i+N/2 \notin I$, where the set of indices $I$ comes from Lemma \ref{overlapmatch}. Let the resulting graph be $\Gamma^2_N$.

\begin{lemma}\label{hall} There exist matchings $M^2_S \subset \Lambda_X$ and $M^2_T \subset \Lambda_Y$ that saturate $S$ and $T$ respectively, satisfy $|M^2_S| = |M^2_T| = |S| = |T|$, and induce a matching also in the graph $\Gamma^2_N$.
\end{lemma}
\begin{proof} We will only show how to construct $M^2_S$; the construction of $M^2_T$ is completely analogous. It suffices to verify the Hall's condition in corresponding part of the graph $\Gamma^2_N$. Let $S_0 \subset S$ and let $|S| = \alpha N$, $|S_0| = \gamma N$.

We distinguish three cases:

\begin{itemize}

\item $\gamma \leq C_3$: Such subsets $S_0$ satisfy the Hall's condition due to (b) in Theorem \ref{mainthm}.

\item $C_3 \leq \gamma = \alpha$: The relative size of the neighborhood of $S_0$ is at least
$$o(\alpha) + \frac12 (e(\alpha) - o(\alpha)) = \frac12(e(\alpha)+o(\alpha)).$$
For showing this is at least $\alpha$ on $[C_3,1]$, it suffices (due to linearity) to verify it for $\alpha = C_3$, $\alpha=C_5$, $\alpha = 1$. The first case follows from (\ref{numeq2}) and (\ref{numeq3}), the second is due to $C_4-C_3 = C_6 - C_5$ from (\ref{numeq4}) the same as the first one and finally, the last one is immediate.

\item $C_3 \leq \gamma < \alpha$: Using the matching from Lemma \ref{overlapmatch} there are at least $(o(\alpha)+\gamma-\alpha)N$ vertices of $S_0$ matched to a vertex in the set of overlaps $X'_I$. Therefore, the relative size of the neighborhood is at least
$$o(\alpha)+\gamma-\alpha + \frac12 (e(\gamma) - (o(\alpha)+\gamma-\alpha)),$$
therefore it suffices (after simple manipulation) to prove
$$e(\gamma)-\gamma \geq \alpha - o(\alpha).$$
Since for the currently considered $\alpha$ and $\gamma$, we have $e(\gamma) - \gamma \geq e(\alpha)-\alpha$ (again due to decreasing slopes from (\ref{numeq4})), the previously established $\frac12\left(e(\alpha)+o(\alpha)\right) \geq \alpha$ gives the conclusion.
\end{itemize}
\end{proof}

It is shown in \cite{AlonC03} that from the matchings $M^1_S$, $M^1_T$, $M^2_S$, $M^2_T$ one can construct matchings $M^*_S$ and $M^*_T$ that satisfy both Lemma \ref{hall} and the conditions (a) and (b) of Lemma \ref{overlapmatch}. These matchings are easily seen to satisfy the conditions of Lemma \ref{suffcond} and this concludes our proof.

\end{proof}

\section{Expanders and Pair-expanders with Fractional Degree} \label{sec:expan}

In order to prove Theorem \ref{expexist} we use a probabilistic argument strongly following the ideas from \cite{Bas81}. The optimization carried out in the previous sections does not guarantee the existence of suitable expanders with degree $5$ which would improve the degree 6 used in \cite{Scho2006}. Therefore we introduce expanders with fractional degree and develop the criteria for their existence.

For this entire section, let $H(x) = -x\log x - (1-x)\log(1-x)$ be the binary
entropy function with $H(0) = H(1) = 0$. We use this function for asymptotic estimates of binomial coefficients.

Finally, let us from now on use the convention that ${n \choose k} = 0$ for $k< 0$ and $k>n$.

\begin{lemma}\label{stirling} (a) There exists $n_0\in\mathbb N$ such that for any integers $k,n$ with $0\le k\le n$ and $n\ge n_0$ it holds that
\begin{equation}
\left|\frac{1}{n} \log{n \choose k} - H\left( \frac{k}{n} \right)\right| < 2\cdot\frac{\log n}{n}
\end{equation}
(b) For any $\epsilon_1,\epsilon_2$ with $0<\epsilon_1<\epsilon_2<1$ there exists $n_0\in\mathbb N$ such that for any $\alpha\in[\epsilon_1,\epsilon_2]$ and any integer $n\ge n_0$ we have
\begin{equation}
\left|\frac{1}{n} \log{n \choose \lfloor\alpha n\rfloor} - H\left( \alpha \right)\right| < 3\cdot\frac{\log n}{n}
\end{equation}
\end{lemma}
\begin{proof} 
\noindent {\bf Part (a)~~}
For $k=0$ and $k=n$ the existence of such $n_0$ can be checked directly; we thus assume that $0<k<n$.
We will use the Stirling estimates for factorials of positive integers $m>0$:
\begin{align*}
\sqrt{2\pi}\ m^{m+1/2}e^{-m} \le m! &\le e\ m^{m+1/2}e^{-m} 
\end{align*}
This implies that
\begin{equation*}
\log  m! - m \log m = -m\log e + \frac{1}{2}\log m + C_m, \quad C_m \in [const_1,const_2].
\end{equation*}
Combining these relations for $m=n$, $m=k$ and $m=n-k$ (the last two with the ``minus'' sign) and dividing by $n$ gives
\begin{equation*}
\frac{1}{n}\log  \binom{n}{k} -  H\left(\frac{k}{n}\right)  = \frac{1}{2n}\left[ \log n - \log k - \log (n-k)\right] + \frac{C_{nk}}{n},
\end{equation*}
where $C_{nk} \in [const'_1,const'_2]$.
This implies part (a) of the lemma.

\noindent {\bf Part (b)~~} Fix $\epsilon'_1\in(0,\epsilon_1)$.
Since function $H(\cdot)$ has a bounded derivative on $[\epsilon'_1,\epsilon_2]$,
we have $|H(\alpha)-H(\alpha')|\le const\cdot |\alpha-\alpha'|$ for any $\alpha,\alpha'\in[\epsilon'_1,\epsilon_2]$
(where the constant depends on $\epsilon'_1,\epsilon_2$). We will take $\alpha'=\lfloor\alpha n\rfloor/n$ (which belongs to $[\epsilon'_1,\epsilon_2]$ for a sufficiently large $n$), then $|\alpha-\alpha'|\le 1/n$ and so
$|H(\alpha)-H(\alpha')|\le const/n$. Applying part (a) to $k=\lfloor\alpha n\rfloor$ then gives the claim.
\end{proof}

Given the set $L$ of left vertices $\{l_1, \dots, l_n\}$, the set $R$ of right vertices $\{r_1,\dots, r_N\}$, and $0\leq \delta \leq 1$ we form a random bipartite graph $G(N, d, \delta)$ as follows. First, we overlay $d$ random permutation graphs and then we draw edges $(l_i, r_i)$ for all positive integers $i$ for which $i \leq  \lfloor \delta N \rfloor$.

We prove that the graph $G(N, d, \delta)$ satisfies certain expansion and pair-expansion properties with high probability.

For the case of pair-expansion we restrict ourselves to the case $\delta \leq \frac12$ as it allows us to prove better constants.

\begin{prop}\label{pairexpansion} Consider some constants $d\in\mathbb N$, $\delta\in[0,1/2]$, $\epsilon_1,\epsilon_2$ with $0<\epsilon_1<\epsilon_2<1$,  and $\gamma > \frac12$ such that $2\epsilon_2\gamma < 1$. Suppose that 
\begin{equation}
d > 1+\gamma \cdot \frac{1-\epsilon_1}{1-2\gamma\epsilon_1}\label{eq:pairexp1}
\end{equation}
and for any $\alpha\in(0,\epsilon_2]$ 
\begin{equation} d + p_\alpha > \frac{H(\alpha)+H(\alpha \beta)}{H(\alpha) - H(1 / \beta)\alpha \beta}\label{eq:pairexp2}
\end{equation}
where $p_\alpha$ satisfies the following:
\begin{enumerate}[(i)]
\item If $\alpha\in(0,\epsilon_1)$ then $p_\alpha=0$.
\item If $\alpha\in[\epsilon_1,\epsilon_2]$ then
\begin{align}
H(\alpha)(1-p_{\alpha})+2\gamma p_\alpha \alpha H\left(\frac{1}{2\gamma}\right) +H(y)> \nonumber \\
\delta H\left(\frac{y}{\delta} \right) + (1-\delta)H\left(\frac{\alpha-y}{1-\delta}\right) + \gamma \alpha H\left(\frac{y}{\gamma\alpha} \right) \label{eq:pairexp3}
\end{align}
\end{enumerate}
for any $y \in [0,\gamma \alpha] \cap [\alpha+\delta-1,\delta]$ (or $c_\alpha=0$ if $\delta=0$).

Then graph $G(N,d,\delta)$ is an $(N,\alpha,\gamma)$-pair-expander for any $\alpha\in[0,\epsilon_2]$ with probability $1-o(1)$.

\end{prop}
Here and below probability $1-o(1)$ is viewed as a function of $N$. It is thus strictly positive for a sufficiently large $N$.

\begin{prop}\label{expansion} Consider some constants $d\in\mathbb N$, $\delta\in[0,1]$, $\epsilon_1,\epsilon_2$ with $0<\epsilon_1<\epsilon_2<1$  and a piecewise linear function $e(\alpha)$ on $[0,1]$ satisfying $\alpha<e(\alpha)<1$ for any $\alpha\in(0,1)$. Suppose that for any $\alpha\in(0,1)$ holds 
\begin{equation}
d + c_\alpha > \frac{H(\alpha) + H(e(\alpha))}{H(\alpha) - H\left( \frac{\alpha}{e(\alpha)} \right) e(\alpha)},\label{eq:exp1}
\end{equation}
where $c_\alpha$ satisfies the following:
\begin{enumerate}[(i)]
\item If $\alpha\in(0,\epsilon_1)\cup(\epsilon_2,1)$ then $c_\alpha=0$.
\item  If $\alpha\in[\epsilon_1,\epsilon_2]$ then
\end{enumerate}
\begin{eqnarray} H(\alpha)(1-c_\alpha)+c_\alpha e(\alpha) H\left( \frac{\alpha}{e(\alpha)} \right) + H(y) > \\
\delta H\left(\frac{y}{\delta}\right) + (1-\delta)H\left(\frac{\alpha-y}{1-\delta} \right)+ e(\alpha) H\left( \frac{y}{ e(\alpha)} \right)\label{eq:exp2}
\end{eqnarray}
for any $y \in [0,\alpha] \cap [\alpha+\delta-1,\delta]$ (or $c_\alpha=0$ if $\delta=0$).
Moreover, suppose that
$$d > 2+e'(0)\qquad \text{and} \qquad d> 1+\frac{2}{e'(1)}$$

Then graph $G(N,d,\delta)$ is an $(N,\alpha,e(\alpha))$-expander for any $\alpha\in[0,1]$ with probability $1-o(1)$.
\end{prop}

Now we show that Propositions \ref{pairexpansion} and \ref{expansion} imply Theorem \ref{expexist}.

\begin{proof}[Proof of Theorem \ref{expexist}] For part (a) we use Proposition \ref{expansion} with $d=5$, $\delta = 0.325$, $\epsilon_1 = 0.21$, $\epsilon_2 = 0.48$ and $c_\alpha = 0.18$ for $\alpha \in [0.21,0.48]$. Inequality (\ref{eq:exp1}) can be checked directly and for inequality (\ref{eq:exp2}) we give a computer-aided proof in Appendix B.

Part (b) is ensured similarly from Proposition \ref{pairexp}. We take $d=5$, $\delta = 0.325$, $\epsilon_1 = 0.3$, $\epsilon_2 = 0.3322$, $\gamma = 1$, and $p_\alpha = 0.45$ for $\alpha \in [0.3,0.3322]$.
Inequalities (\ref{eq:pairexp1}) and (\ref{eq:pairexp2}) can again be checked directly and for inequality (\ref{eq:pairexp3}) we give a computer-aided proof in Appendix B.

All in all, the random graph $G(N,5,0.325)$ both (a) and (b) with probability at least $1-o(1)-o(1)$, which is $1-o(1)$. In particular, this probability
is strictly positive for a sufficiently large $N$.
\end{proof}

\section{Proof of Proposition \ref{pairexpansion}}

First, we estimate the probability of the pair-expansion property and then we decompose Proposition \ref{pairexpansion} naturally into its fractional and non-fractional part.

\begin{lemma}\label{pairprob} Let $d \in \mathbb{N}$, $0 \leq \delta < \frac12$, $k \leq N$, and $G = G(N,d,\delta)$ with $N$ left vertices $L$ and $N$ right vertices $R$. Then the probability that some $U \subset L$, $|U| = k$ fails to have at least $m$ ($k/2 < m < N/2$) neighboring pairs is at most
$${N/2 \choose m-1}  \left( \frac{{2m - 2 \choose k}}{{N \choose k}} \right)^d  \sum_{i=0}^{m-1} {\lfloor \delta N \rfloor \choose i}{N- \lfloor \delta N \rfloor \choose k-i} { m - 1 \choose i} \bigg/ {N \choose i}
$$
which in the case $\delta = 0$ reduces to
$${N \choose k} {N/2 \choose m-1}  \left( \frac{{2m - 2 \choose k}}{{N \choose k}} \right)^d.$$
\end{lemma}
\begin{proof} Let us first fix a set $U \subset L$ of size $k$ and compute the probability it fails in the pair-expansion. That happens if and only if there exists $V \subset R$ formed by $m-1$ pairs such that the neighbours of $U$ lie entirely in $V$. Choose $V \subset R$ consisting of $m-1$ pairs randomly. For the $d$ complete permutations the probability is
$$
\left( \frac{{2m-2 \choose k}}{{N \choose k}} \right)^d.
$$
Let the probability concerning the extra $\lfloor \delta N \rfloor$ edges be $p_U$.
From the union bound over subsets $V$ and also over subsets $U$ of size $k$, we upper bound the probability of failing in pair-expansion as
$$
{N/2 \choose m-1}  \left( \frac{{2m - 2 \choose k}}{{N \choose k}} \right)^d  \sum_U p_U.
$$
The sum can be upper-bounded using the union bound over the possible cardinalities of $U \cap \{l_i| i = 1,\dots, \lfloor \delta N \rfloor \}$ as follows
$$
\sum_U p_U \leq \sum_{i=0}^{\min(m-1,k)} {\lfloor \delta N \rfloor \choose i}{N- \lfloor \delta N \rfloor \choose k-i} { m - 1 \choose i} \bigg/ {N \choose i}
$$
where we use the fact that $\lfloor \delta N \rfloor$ edges connect disjoint pairs as $\delta < 1/2$. This proves the first part of the claim and for the second one we may for example observe that $p_U = 1$ for any $U$ when $\delta=0$.
\end{proof}

\begin{prop}\label{pairexp} Let $d \in \mathbb{N}$, $\alpha \in (0,1)$, $\gamma > 1/2$, and $2\alpha \gamma < 1$. Then the graph $G(N,d,0)$ is a $(\alpha',\gamma)$-pair-expander for each $0 \leq \alpha' \leq \alpha$ with probability $1-o(1)$ if
\begin{equation}
d > \frac{H(\alpha) + \frac12 H(2\gamma\alpha)}{H(\alpha) - 2\gamma\alpha H\left(\frac{1}{2\gamma}\right)}\quad \text{and} \quad d > 1+\gamma \cdot \frac{1-\alpha}{1-2\gamma\alpha}\;\;.\label{eq:pairexpnonfrac}
\end{equation}
\end{prop}
\begin{proof} For sets of size $k$ where $1 \leq k \leq \lfloor \alpha N \rfloor$ the probability of failing in pair-expansion is by Lemma \ref{pairprob} at most
$${N \choose k} {N/2 \choose \lfloor \gamma k \rfloor -1}  \left( \frac{{2\lfloor \gamma k \rfloor - 2 \choose k}}{{N \choose k}} \right)^d$$
and after using the union bound over values of $k$, the total probability of failing is at most
$$
\sum_{k=1}^{\lfloor \alpha N \rfloor} {N \choose k}{N/2 \choose \lfloor \gamma k \rfloor -1}  \left( \frac{{2\lfloor \gamma k \rfloor - 2 \choose k}}{{N \choose k}} \right)^d.
$$

We will show that each summand is (significantly) smaller than $1/(\alpha N)$ for large $N$. Let us distinguish two cases.

\begin{enumerate}[(a)]
\item $k \leq \varepsilon N$: Note that (\ref{eq:pairexpnonfrac}) implies also $d > 1+\gamma$. Then standard estimates on binomial coefficients
$$\binom{n}{k} \leq \left(\frac{ne}{k}\right)^k, \qquad \binom{n}{k} \Big/ \binom{m}{k} \leq \left(\frac{n}{m}\right)^k \quad \text{for} \, n \leq m
$$
give
\begin{align*} \alpha N &{N \choose k}{N/2 \choose \lfloor \gamma k \rfloor -1}  \left( \frac{{2\lfloor \gamma k \rfloor - 2 \choose k}}{{N \choose k}} \right)^{d}\\ &\leq \alpha N \left( \frac{eN}{k} \right)^k  \left(\frac{eN}{2(\gamma k-1)} \right)^{\gamma k-1} \left( \frac{2\gamma k}{N} \right)^{kd} \\
&\leq \frac{2\alpha\gamma k}{e}  \left( \frac{eN}{k} \right)^k  \left(\frac{eN}{2(\gamma k-1)} \right)^{\gamma k} \left( \frac{2\gamma k}{N} \right)^{kd} \\
&= \frac{2\alpha\gamma k}{e} \left( 2\gamma e^{1+\gamma} \left( 2\gamma \cdot \frac{k}{N} \right)^{d-\gamma-1} \left(1+\frac{1}{\gamma k -1}\right)^\gamma  \right)^k \\
&\leq C_1k(C_2 \varepsilon^{d-\gamma-1})^k,
\end{align*}
for some $C_1, C_2$ independent from $N$ and $k$.
By choosing suitable constant $\varepsilon > 0$ this can be made arbitrarily small for all $k \leq \varepsilon N$ if we make use of $d > 1+\gamma$.

\item $\varepsilon N  < k \leq \alpha N$: As both $N$ and $k$ are now arbitrarily large, we may use the Stirling estimates and obtain
\begin{align*} 
&\alpha N {N \choose k}{N/2 \choose \lfloor \gamma k \rfloor -1}  \left( \frac{{2\lfloor \gamma k \rfloor - 2 \choose k}}{{N \choose k}} \right)^{d} \leq \\
&\exp \!\left( \! N \! \left( \! H(x) + \frac12 H(2\gamma x) +2d\gamma x H\left(\frac{1}{2\gamma} \right) - dH(x)  \right)\! +\!O(\log N)\! \right)
\end{align*}
where $x = k/N$.
It is straightforward to verify that the function
$$F(x) = H(x) + \frac12 H(2\gamma x) +2d\gamma x H\left(\frac{1}{2\gamma} \right) - dH(x) $$ is convex on $[\varepsilon,\alpha]$ if
$$d > 1+\gamma \cdot \frac{1-\alpha}{1-2\gamma\alpha}$$
which we ensured in (\ref{eq:pairexpnonfrac}).
Therefore $F$ attains its maximum on $[\varepsilon, \alpha]$ at its endpoints. We easily get $F(\varepsilon) < 0$ if $d > 1+\gamma$ and $F(\alpha) < 0$ if
$$
d > \frac{H(\alpha) + \frac12 H(2\gamma\alpha)}{H(\alpha) - 2\gamma\alpha H\left(\frac{1}{2\gamma}\right)}.
$$
\end{enumerate}
This implies the result.
\end{proof}

\begin{prop} Let $d \in \mathbb{N}$, $0 \leq \delta < 1/2$, $0< \epsilon_1 < \epsilon_2 < 1$, $\gamma > 1/2$, and $2\epsilon_2 \gamma < 1$. Then the graph $G(N,d,\delta)$ is a $(N,\alpha,\gamma)$-pair-expander for every $\alpha \in [\epsilon_1, \epsilon_2]$ with probability $1-o(1)$ if
\begin{equation} d+p_\alpha > \frac{H(\alpha) + \frac12 H(2\gamma\alpha)}{H(\alpha) - 2\gamma\alpha H\left(\frac{1}{2\gamma}\right)}, \label{eq:pairexpineq}
\end{equation}
for each $\alpha \in [\epsilon_1, \epsilon_2]$, where $p_\alpha$ is a number for which the following inequality holds:
\begin{align}
H(\alpha)(1-p_\alpha)+2\gamma p_\alpha \alpha H\left(\frac{1}{2\gamma}\right) +H(y)> \nonumber \\ \delta H\left(\frac{y}{\delta} \right) + (1-\delta)H\left(\frac{\alpha-y}{1-\delta}\right) + \gamma \alpha H\left(\frac{y}{\gamma\alpha} \right) \label{eq:pairexpineq2}
\end{align}
for any $\alpha \in [\epsilon_1, \epsilon_2]$ and any $y \in [0,\gamma \alpha] \cap [\alpha+\delta-1,\delta]$.

\end{prop}
\begin{proof} 
For sets of size $k$ where $\lfloor \epsilon_1 N \rfloor \leq k \leq \lfloor \epsilon_2 N \rfloor$ the probability of failing in pair-expansion is by Lemma \ref{pairprob} at most
$${N/2 \choose \lfloor \gamma k \rfloor-1}  \left( \frac{{2\lfloor \gamma k \rfloor - 2 \choose k}}{{N \choose k}} \right)^d  \sum_{i=0}^{\lfloor \gamma k \rfloor -1} R_i$$
where
$$R_i = {\lfloor \delta N \rfloor \choose i}{N- \lfloor \delta N \rfloor \choose k-i} { \lfloor \gamma k \rfloor - 1 \choose i} \bigg/ {N \choose i}.
$$
From the union bound over feasible values of $k$ the total probability of failing in expansion is at most
$$\sum_{k=\lfloor \epsilon_1 N \rfloor}^{\lfloor \epsilon_2 N \rfloor}\left( {N/2 \choose \lfloor \gamma k \rfloor-1}  \left( \frac{{2\lfloor \gamma k \rfloor - 2 \choose k}}{{N \choose k}} \right)^d  \sum_{i=0}^{\lfloor \gamma k \rfloor -1} R_i \right).
$$
We will show that there is $c > 0$ such that for sufficiently large $N$ ($N > N_0$) each summand is at most $e^{-cN}$. Since the number of summands is linear in $N$, the conclusion will follow.

First note that both inequalities (\ref{eq:pairexpineq}) and (\ref{eq:pairexpineq2}) are strict and hold over compact sets so they can both be strengthened by some $\varepsilon > 0$ (independent of $\alpha$).

We decompose the inequality into two estimates. 

For the first one let
$$L =  {N \choose k} \left(\frac{{2 \lfloor \gamma k \rfloor - 2 \choose k}}{{N \choose k}} \right)^{p_\alpha},\quad R = \sum_{i=0}^{\lfloor \gamma k \rfloor -1} R_i.$$
We claim that $R/L < e^{-c_1N}$ for some constant $c_1 > 0$ and $N > N_0$ where $c_1$ and $N_0$ are both independent of $k$. Again it suffices to prove that for some $c_2>0$ and $N > N_0$ (both independent of $k$) we have $R_i/L < e^{-c_2N}$ for all $i$.

To this end, we use the Stirling estimates to see that for $N > N_0$
\begin{align*}
\frac{1}{N} \log (R_i/L) &<
\delta H\left(\frac{y}{\delta} \right) + (1-\delta)H\left(\frac{\alpha-y}{1-\delta}\right) \\ + \gamma \alpha H\left(\frac{y}{\gamma\alpha} \right)& -\left(H(\alpha)(1-p)+2\gamma p \alpha H\left(\frac{1}{2\gamma}\right) +H(y)\right) + \frac{\varepsilon}{2}
\end{align*}
where $\alpha = k/N$ and $y = i/N$. Moreover, by Lemma \ref{stirling} this $N_0$ does not depend on $k$ and $i$.
Using (\ref{eq:pairexpineq2}) strengthened by $\varepsilon$, we finally obtain
that for $N> N_0$ we have
$$\frac{1}{N} \log (R_i/L) < -\frac{\varepsilon}{2}$$ for all $i$, where $N_0$ is independent of $k$. This proves the estimate.

Applying this estimate, we are left to prove that for some $c> 0$ and $N>N_0$
$${N \choose k}{N/2 \choose \lfloor \gamma k \rfloor -1}  \left( \frac{{2 \lfloor \gamma k \rfloor - 2 \choose k}}{{N \choose k}} \right)^{d+p_\alpha} < e^{-cN}$$
holds for all admissible values of $k$.
Again we employ the Stirling estimates to upper-bound the left-hand side by $e^{c_1N}$, where
$$c_1 < H(\alpha) + \frac12 H(2\gamma \alpha) +(d+p_\alpha)\left(2\gamma \alpha H\left(\frac{1}{2\gamma} \right) - H(\alpha)\right) + \frac{\varepsilon}{2} < -\frac{\varepsilon}{2}$$
for $N > N_0$ with $N_0$ independent of $k$ (due to Lemma \ref{stirling}) and where we used the strengthened (\ref{eq:pairexpineq}) in the second estimate.

This concludes the proof.
\end{proof}

It is easy to see that the previous two propositions immediately imply Proposition \ref{pairexpansion}.

\section{Proof of Proposition \ref{expansion}}

The proof of Proposition \ref{expansion}, to which this section is devoted, goes along the same lines as the one in the previous section.

\begin{lemma}\label{expansionprob} Let $d \in \mathbb{N}$, $0 \leq \delta < 1$, $1 \leq k \leq N$, and $G = G(N,d,\delta)$ with $N$ left vertices $L$ and $N$ right vertices $R$. Then the probability that some $U \subset L$, $|U| = k$ fails to have at least $m$ ($1 \leq m \leq N$) neighboring pairs is at most
$${N \choose m -1}  \left( \frac{{m - 1 \choose k}}{{N \choose k}} \right)^d  \sum_{i=0}^{k} {\lfloor \delta N \rfloor \choose i}  {N - \lfloor \delta N \rfloor \choose k - i}  {m -1 \choose i} \bigg/ {N \choose i}$$
which in the case $\delta = 0$ reduces to
$${N \choose k} {N \choose m -1}  \left( \frac{{m - 1 \choose k}}{{N \choose k}} \right)^d.$$
\end{lemma}
\begin{proof} Let us first fix a set $U \subset L$ of size $k$ and compute the probability it fails in the expansion. That happens if and only if there exists $V \subset R$ formed by $m-1$ vertices such that the neighbours of $U$ lie entirely in $V$. Choose $V \subset R$ consisting of $m-1$ vertices randomly. For the $d$ complete permutations the probability is
$$
\left( \frac{{m-1 \choose k}}{{N \choose k}} \right)^d.
$$
Let the probability concerning the extra $\lfloor \delta N \rfloor$ edges be $p_U$.
From the union bound over subsets $V$ and also over subsets $U$ of size $k$, we upper bound the probability of failing in expansion as
$$
{N \choose m-1}  \left( \frac{{m-1 \choose k}}{{N \choose k}} \right)^d  \sum_U p_U.
$$
The sum can be upper-bounded using the union bound over the possible cardinalities of \linebreak $U \cap \{l_i\mid i = 1,\dots, \lfloor \delta N \rfloor \}$ as follows
$$
\sum_U p_U \leq \sum_{i=0}^{k} {\lfloor \delta N \rfloor \choose i}{N- \lfloor \delta N \rfloor \choose k-i} { m - 1 \choose i} \bigg/ {N \choose i}.
$$
This proves the first part of the claim and for the second one we may for example observe that $p_U = 1$ for any $U$ when $\delta=0$.
\end{proof}

Next, we will analyze three cases: (i) $\alpha$ is far from $0$ and $1$; (ii) $\alpha$ is close to 0; (iii) $\alpha$ is close to 1.
(In the previous section we needed to worry only about the first two). We will start with the first case.

\begin{prop}\label{expansionclosedint} Let $d \in \mathbb{N}$, $0 \leq \delta < 1$, $0< \epsilon_1 < \epsilon_2 < 1$, and let $e(\alpha)$ be a continuous function on $[\epsilon_1,\epsilon_2]$ for which $\alpha < e(\alpha) < 1$ for all $\alpha \in [\epsilon_1, \epsilon_2]$. Then the graph $G(N,d,\delta)$ is a $(N,\alpha,e(\alpha))$-expander for every $\alpha \in [\epsilon_1, \epsilon_2]$ with probability $1-o(1)$ if one of the two following conditions holds:
\begin{enumerate}[(i)]
\item $\delta = 0$ and
\begin{equation}
d > \frac{H(\alpha) + H(e(\alpha))}{H(\alpha) - H\left( \frac{\alpha}{e(\alpha)} \right) e(\alpha)}, \label{eq:expansionineq1}
\end{equation}
for each $\alpha \in [\epsilon_1, \epsilon_2]$
\item $\delta > 0$ and
\begin{equation} d+c_\alpha > \frac{H(\alpha) + H(e(\alpha))}{H(\alpha) - H\left( \frac{\alpha}{e(\alpha)} \right) e(\alpha)}, \label{eq:expansionineq2}
\end{equation}
for each $\alpha \in [\epsilon_1, \epsilon_2]$, where $c_\alpha$ is a number for which the following inequality holds:
\begin{align}
H(\alpha)(1-c_\alpha)+c_\alpha e(\alpha) H\left( \frac{\alpha}{e(\alpha)} \right) + H(y) >  \nonumber\\
\delta H\left(\frac{y}{\delta}\right) + (1-\delta)H\left(\frac{\alpha-y}{1-\delta} \right)+ e(\alpha) H\left( \frac{y}{ e(\alpha)} \right)\label{eq:expansionineq3}
\end{align}
for any $\alpha \in [\epsilon_1, \epsilon_2]$ and any $y \in [0,\alpha] \cap [\alpha+\delta-1,\delta]$.
\end{enumerate}
\end{prop}
\begin{proof} Let us begin with the first part and assume $\delta = 0$.

Then for sets of size $k$ where $\lfloor \epsilon_1 N \rfloor \leq k \leq \lfloor \epsilon_2 N \rfloor$ the probability of failing in expansion is by Lemma \ref{expansionprob} at most
$${N \choose k} {N \choose \lceil e(\alpha) N \rceil -1}  \left( \frac{{\lceil e(\alpha) N \rceil-1 \choose k}}{{N \choose k}} \right)^d$$
where $\alpha = k/N$,
and after using the union bound over values of $k$, the total probability of failing is at most
\begin{equation}
\sum_{k=\lfloor \epsilon_1 N \rfloor}^{\lfloor \epsilon_2 N \rfloor} {N \choose k} {N \choose \lceil e(\alpha) N \rceil -1}  \left( \frac{{\lceil e(\alpha) N \rceil-1 \choose k}}{{N \choose k}} \right)^d. \label{eq:expansionprobsum}
\end{equation}

We will show that there is $c > 0$ such that for sufficiently large $N$ ($N > N_0$) each summand is at most $e^{-cN}$. Since the number of summands is linear in $N$, the conclusion will follow.

First note that the inequality (\ref{eq:expansionineq1}) is strict and holds over a compact set so it can be strengthened by some $\varepsilon > 0$ (independent of $\alpha$).

Again we employ the Stirling estimates to upper-bound each term of (\ref{eq:expansionprobsum}) by $e^{c_1N}$, where
$$c_1 < H(\alpha) + H(e(\alpha)) +d\left( H(\alpha) - H\left( \frac{\alpha}{e(\alpha)} \right) e(\alpha)\right) + \frac{\varepsilon}{2} < -\frac{\varepsilon}{2}$$
for $N > N_0$ with $N_0$ independent of $k$ (due to Lemma \ref{stirling}) and where we used the strengthened (\ref{eq:expansionineq1}) in the second estimate.

This finishes the proof of the case $\delta = 0$.

\medskip

Now let $\delta > 0$. Using again Lemma \ref{expansionprob} and the union bound over $k$, we get that the total probability of failing in expansion is at most

$$\sum_{k=\lfloor \epsilon_1 N \rfloor}^{\lfloor \epsilon_2 N \rfloor} {N \choose \lceil e(\alpha) N \rceil -1}  \left( \frac{{\lceil e(\alpha) N \rceil - 1 \choose k}}{{N \choose k}} \right)^d  \sum_{i=0}^{k} R_i,$$
where
$$
R_i = {\lfloor \delta N \rfloor \choose i}  {N - \lfloor \delta N \rfloor \choose k - i}  {\lceil e(\alpha) N \rceil -1 \choose i} \bigg/ {N \choose i}.
$$
We will show that there is $c > 0$ such that for sufficiently large $N$ ($N > N_0$) each summand is at most $e^{-cN}$. Since the number of summands is linear in $N$, the conclusion will follow.

Note that (\ref{eq:expansionineq3}) is strict and holds over a compact set so it can be strengthened by some $\varepsilon > 0$ (independent of $\alpha$).

Let
\begin{align*}
L &=  {N \choose k} \left( \frac{{\lceil e(\alpha) N \rceil - 1 \choose k}}{{N \choose k}} \right)^{c_\alpha},\\
R &= \sum_{i=0}^{\lfloor \gamma k \rfloor -1} R_i.
\end{align*}
We claim that $R/L < e^{-c_1N}$ for some constant $c_1 > 0$ and $N > N_0$ where $c_1$ and $N_0$ are both independent of $k$. Again it suffices to prove that for some $c_2>0$ and $N > N_0$ (both independent of $k$) we have $R_i/L < e^{-c_2N}$ for all $i$.

To this end, we use the Stirling estimates to see that for $N > N_0$
\begin{align*}
\frac{1}{N} \log (R_i/L) <&\\
&\delta H\left(\frac{y}{\delta}\right) + (1-\delta)H\left(\frac{\alpha-y}{1-\delta} \right)+ e(\alpha) H\left( \frac{y}{ e(\alpha)} \right) \\
&- \left( H(\alpha)(1-c_\alpha)+c_\alpha e(\alpha) H\left( \frac{\alpha}{e(\alpha)} \right) + H(y) \right)   + \frac{\varepsilon}{2}
\end{align*}
where $\alpha = k/N$ and $y = i/N$. Moreover, by Lemma \ref{stirling} this $N_0$ does not depend on $k$ and $i$.
Using (\ref{eq:expansionineq2}) strengthened by $\varepsilon$, we finally obtain
that for $N> N_0$ we have
$$\frac{1}{N} \log (R_i/L) < -\frac{\varepsilon}{2}$$ for all $i$, where $N_0$ is independent of $k$. This proves the estimate.

Applying this estimate, we are left to prove that for some $c> 0$ and $N>N_0$
$${N \choose k} \left( \frac{{\lceil e(\alpha) N \rceil - 1 \choose k}}{{N \choose k}} \right)^{d+c_\alpha} < e^{-cN}$$
holds for all admissible values of $k$.
Here we may join the proof of the first part of this proposition with $\delta + c_\alpha$ playing the role of $\delta$.
\end{proof}

The next proposition analyzes the case when $\alpha$ is close to $0$.
\begin{prop}\label{closetozero} Let \footnote{This situation was treated already in \cite{Bas81} leading to a better sufficient condition $d>1+\beta$.
However, in the proof an incorrect estimate $n\binom{n}{k} \leq k\left(\frac{en}{k}\right)^k$ was used (see their inequality (b) at the bottom of page 83).
Here, we derive a weaker version which is still applicable in our case.} $d \in \mathbb{N}$ and $\beta > 1$. Then there exists $\varepsilon > 0$ such that the graph $G(N,d,0)$ is a $(N,\alpha,\alpha\beta)$-expander for every $\alpha \in [0, \varepsilon]$ with probability $1-o(1)$ if 
$$d > 2+ \beta.$$
\end{prop}
\begin{proof} The probability we want to upper-bound is by Lemma \ref{expansionprob} and after applying the union bound over acceptable values of $k$ at most
\begin{equation}
\sum_{k=1}^{\lfloor \varepsilon N \rfloor} {N \choose k} {N \choose \lceil k\beta \rceil -1}  \left( \frac{{\lceil k\beta \rceil-1 \choose k}}{{N \choose k}} \right)^d.
\end{equation}

We will prove that each term can be made (significantly) smaller than $1/N$.

Then standard estimates on binomial coefficients
$$\binom{n}{k} \leq \left(\frac{ne}{k}\right)^k, \qquad \binom{n}{k} \Big/ \binom{m}{k} \leq \left(\frac{n}{m}\right)^k \quad \text{for} \, n \leq m
$$
give
\begin{align*} N {N \choose k} &{N \choose \lceil \beta k \rceil -1}  \left( \frac{{\lceil \beta k\rceil-1 \choose k}}{{N \choose k}} \right)^d \\ &\leq  N \left( \frac{eN}{k} \right)^k  \left(\frac{eN}{\beta k} \right)^{\beta k} \left( \frac{\beta k}{N} \right)^{kd} \\
&\leq  k \left( \frac{eN}{k} \right)^k  \left(\frac{eN}{\beta k} \right)^{\beta k+k} \left( \frac{\beta k}{N} \right)^{kd} \\
&= k \left(\beta^{d-\beta-1} e^{\beta+2} \left(\frac{k}{N} \right)^{d-\beta-2} \right)^k \\
&\leq k(C \varepsilon^{d-\beta-2})^k,
\end{align*}
for some $C$ independent from $N$ and $k$.
By choosing suitable constant $\varepsilon > 0$ this can be made arbitrarily small for all $k \leq \varepsilon N$ if we make use of $d > 2+\beta$.
\end{proof}

Finally, we need to consider the case when $\alpha$ is close to $1$. We will need the following fact.
\begin{lemma}\label{binombound} If for $n,k,m \in \mathbb{N}$ holds $n+m > 2k >2m$, then
$${n \choose k-m} \bigg/ {n \choose k} \leq \left(\frac{k}{n-k+m} \right)^m.$$
\end{lemma}
\begin{proof}
We recall that $${n \choose k-1} \bigg/ {n \choose k} = \frac{k}{n-k+1}$$ and use it inductively to get
\begin{equation} 
\begin{split}
{n \choose k-m} \bigg/ {n \choose k} & = \frac{k \dotsm (k-m+1)}{(n-k+1)\dotsm (n-k+m)}\\
& = \frac{k}{n-k+m} \dotsm \frac{k-m+1}{n-k+1}\\
& \leq \left( \frac{k}{n-k+m} \right)^m
\end{split}
\end{equation}
where in the last inequality we have used that $k < n-k+m$ and therefore the first fraction provides an upper bound for all others.
\end{proof}

\begin{prop}\label{closetoone} Let $d \in \mathbb{N}$ and $0<c<1$. Then there exists $\varepsilon > 0$ such that the graph $G(N,d,0)$ is a $(N,\alpha,1-c(1-\alpha))$-expander for every $\alpha \in [1-\varepsilon, 1]$ with probability $1-o(1)$ if 
$$d > 1+ \frac{2}{c}.$$
\end{prop}
\begin{proof} Expanding with a function $e(\alpha) = 1-c(1-\alpha)$ implies that sets of size $k$ will expand to size at least $N-\lfloor c(N-k) \rfloor.$

Then an upper bound on the probability of failing in expansion is given by Lemma \ref{expansionprob} and after applying the union bound over acceptable values of $k$ this is
\begin{equation}
\sum_{k=\lfloor (1-\varepsilon)N \rfloor}^{N-1} {N \choose k} {N \choose N-\lfloor c(N-k) \rfloor -1}  \left( \frac{{N-\lfloor c(N-k) \rfloor -1 \choose k}}{{N \choose k}} \right)^d.
\end{equation}

We will prove that each term can be made (significantly) smaller than $1/N$. To this end, let $k' = N-k$, assume that $k' < N/3$ and use standard estimates on binomial coefficients together with Lemma \ref{binombound} to get
\begin{align*}
N& {N \choose k} {N \choose N-\lfloor c(N-k) \rfloor -1}  \left( \frac{{N-\lfloor c(N-k) \rfloor -1 \choose k}}{{N \choose k}}\right)^d \\
&= {N \choose k'} {N \choose \lfloor ck' \rfloor +1}  \left( \frac{{N - \lfloor ck' \rfloor -1\choose k'-\lfloor ck'\rfloor -1}}{{N \choose k'}}\right)^d\\
&\leq N  \left( \frac{Ne}{k'} \right)^{k'}  \left(\frac{Ne}{ck'+1} \right)^{ck'+1}  \left( \frac{{N-\lceil ck' \rceil \choose k'-\lceil ck'\rceil}}{{N \choose k'}}\right)^d \\
&\leq \frac{ck'+1}{e}  \left( \frac{Ne}{k'} \right)^{k'}  \left(\frac{Ne}{ck'+1} \right)^{ck'+k'} \left( \frac{{N \choose k'-\lceil ck'\rceil}}{{N \choose k'}}\right)^d \\
&\leq \frac{ck'+1}{e}  \left( \frac{Ne}{k'} \right)^{k'}  \left(\frac{Ne}{ck'+1} \right)^{ck'+k'} \left( \frac{k'}{N-k'+\lfloor ck' \rfloor}\right)^{ck'd},
\end{align*}
where after writing 
$$\left(\frac{Ne}{ck'+1} \right)^{ck'+k'} = \left( \frac{N}{k'} \right)^{ck'+k'}  \left(\frac{e}{c+\dfrac{1}{k'}} \right)^{ck'+k'}$$
and
$$\left( \frac{k'}{N-k'+\lfloor ck' \rfloor}\right)^{ck'd} = \left( \frac{k'}{N} \right)^{ck'd} \left( \frac{1}{1-\dfrac{k'-\lceil ck' \rceil}{N}}\right)^{ck'd},$$
the entire left-hand side can be upper-bounded by
$$(C_1k'+C_2)\left( \left(\frac{k'}{N}\right)^{dc-c-2} C_3 \right)^{k'}$$
for some $C_1$, $C_2$, $C_3$ independent of $N$ and $k$. By choosing suitable constant $\varepsilon > 0$ this can be made arbitrarily small for all $k' \leq \varepsilon N$ if we make use of $d > 1+ \frac{2}{c}$.
\end{proof}

Now we have all it takes to prove Proposition \ref{expansion}. We fix $0<\epsilon_1 < \epsilon_2 < 1$ and find $\xi > 0$ such that $e(\alpha)$ is linear on $[0,\xi]$ and $[1-\xi,1]$. Using Propositions \ref{closetozero} and \ref{closetoone}, we find $\varepsilon > 0$ such that $\xi >\varepsilon$, $\epsilon_1 > \varepsilon$, $1-\epsilon_2 > \varepsilon$, and such that the expansion is guaranteed for $\alpha \in [0,\varepsilon] \cup [1-\varepsilon,1]$. For the expansion on the intervals $[\varepsilon, \epsilon_1]$ and $[\epsilon_2,1-\varepsilon]$ we employ Proposition \ref{expansionclosedint} with $\delta = 0$ and for the remaining interval $[\epsilon_1,\epsilon_2]$ we also employ Proposition \ref{expansionclosedint} but this time the version for $\delta > 0$. This guarantees the desired expansion for all $\alpha \in [0,1]$.
\hfill\qedsymbol

\nocite{*}
\bibliographystyle{plain}
\bibliography{references}

\newpage

\appendix

\section{}

In this section we consider the following problem.
Let $E_n$ be a bipartite graph with $n$ left vertices $L$ and $n$ right vertices $R$
of an integer degree $d$ obtained as a union of $d$ random permutation graphs.
Fix positive integers $\ell,r\le n$ and subset $U\subseteq L$ of size $\ell$.
We are interested in the probability $p_{\ell r}$
that $U$ has a neighborhood $\Gamma(U)$ of size at most $r$.
The probability that $E_n$ is not an $(n,\ell/n,(r+1)/n)$-expander can then be upper-bounded by $\binom{n}{\ell}\cdot p_{\ell r}$.

For a fixed set $X\subseteq R$ of size $k\le n$ let $p_k$ be the probability that $\Gamma(U)\subseteq X$.
This probability can be easily  computed as
$$p_k = \frac{\binom{n-k}{\ell}}{\binom{n}{\ell}}.$$
Bassalygo~\cite{Bas81} used the following upper bound on $p_{\ell r}$: 
\begin{equation}\label{eq:BassalygoBound}
p_{\ell r}\le \binom{n}r p_r.
\end{equation}
The main result of this section is the following exact expression for $p_{\ell r}$.
\begin{theorem}\label{th:plr}
There holds
\begin{equation}
p_{\ell r}=\sum_{k=0}^m\alpha_k p_k
\end{equation}
where
\begin{equation}\label{eq:p:union:b}
\alpha_k=(-1)^{r-k}\binom{n}{k}\binom{n-k-1}{r-k}.
\end{equation}
\end{theorem}
Our numerical experiments suggest that the estimate \eqref{eq:BassalygoBound}
is very close to the true value of $p_{\ell r}$; the exact value (or rather its version
for the fractional degree) would allow to decrease the
density of a superconcentrator but by a very small amount. Therefore, in the main
part of the paper we used the estimate \eqref{eq:BassalygoBound} for simplicity (or more precisely
its version for the fractional degree).
Theorem \ref{th:plr} is given only as a side result.

To prove this theorem, we will consider a more general problem.
Let $\calS_r=\{X\subseteq R\:|\:|X|\le r\}$. To each $X\in\calS_r$
we will associate an event which will be denoted as $[X]$.
As an example, $[X]$ could be the event that subset $U$ expands entirely into $X$, i.e.\ $\Gamma(U)\subseteq X$.
Theorem \ref{th:plr} will follow from the result below.

\begin{lemma}
Suppose that events $\{[X]\:|\:X\in\calS_r\}$ satisfy the following for some vector ${\bf p}=(p_0,p_1,\ldots,p_r)$:
\begin{subequations}\label{eq:preconditions}
\begin{IEEEeqnarray}{rCll}
\bigwedge_{X\in \calT} [X] &=& [{\bigcap\limits_{X\in\calT}X}] \qquad\quad & \forall \calT\subseteq\calS_r \\
{\mathbb P}([X])&=&p_{|X|} & \forall X\in \calS_r
\end{IEEEeqnarray}
\end{subequations}
Then
\begin{equation}\label{eq:p:union:a}
\mathbb P(\bigvee_{X\in \calS_r} {[X]})=\sum_{k=0}^r\alpha_k p_k
\end{equation}
where coefficients $\alpha_k$ are given by \eqref{eq:p:union:b}.
\end{lemma}

\begin{proof}
By the inclusion-exclusion principle
\begin{align*}
\mathbb P(\bigvee_{X\in \calS_r} {[X]})
&=\sum_{\varnothing\ne\calT\subseteq \calS_r}(-1)^{|\calT|+1}\mathbb P(\bigwedge_{X\in \calT}[X]) \\
&=\sum_{\varnothing\ne\calT\subseteq \calS_r}(-1)^{|\calT|+1}\mathbb P([\bigcap_{X\in \calT}X])
=\sum_{k=0}^r \alpha_k p_k
\end{align*}
where $\alpha_k$ are some constants that depend on $n$ and $r$ (but not on ${\bf p}$).

To compute these constants, we will consider the following example.
Assume $R=\{1,\ldots,n\}$ and consider $n$ Boolean independent variables $Z_1,\ldots,Z_n$ with $\mathbb P(Z_i=0)=q$.
Let $[X]$ be the event that $Z_i=0$ for all $i\in R\setminus X$. Then conditions~\eqref{eq:preconditions} hold
for vector ${\bf p}$ with  $p_i=q^{n-i}$. We also have
\begin{eqnarray*}
\mathbb P(\bigvee_{X\in \calS_r} {[X]})
=\mathbb P(\sum_{i=1}^n Z_i\le r)
&=&\sum_{i=0}^r\binom{n}{i}(1-q)^i q^{n-i} \\
&=&\sum_{i=0}^r\binom{n}{i}\sum_{k=0}^i \binom{i}{k}(-q)^{i-k} q^{n-i} \\
&=&\sum_{i=0}^r\binom{n}{i}\sum_{k=0}^i \binom{i}{k}(-1)^{i-k} q^{n-k} \\
&=&\sum_{k=0}^r \left[ \sum_{i=k}^r(-1)^{i-k}\binom{n}{i} \binom{i}{k} \right] q^{n-k}
\end{eqnarray*}
which must equal $\sum_{k=0}^r\alpha_k p_k=\sum_{k=0}^r\alpha_k q^{n-k}$ for all $q\in[0,1]$. This implies that
\begin{align*}
\alpha_k&=\sum_{i=k}^r(-1)^{i-k}\binom{n}{i} \binom{i}{k}
=\sum_{j=0}^{r-k}(-1)^j\binom{n}{k+j} \binom{k+j}{k} \\
&=\binom{n}{k} \sum_{j=0}^{r-k} (-1)^j \binom{n-k}{j}
\end{align*}
and the sum on the right-hand side can be simplified using the Pascal's rule as
\begin{align*}
\sum_{j=0}^{r-k} (-1)^j \binom{n-k}{j} &= \sum_{j=0}^{r-k} (-1)^j \left( \binom{n-k-1}{j-1}+ \binom{n-k-1}{j} \right)\\ &= (-1)^{r-k}\binom{n-k-1}{r-k}
\end{align*}
where we set $\binom{n-k-1}{-1} = 0$. This establishes~\eqref{eq:p:union:b}.
\end{proof}

\section{}

Here we present computer-aided proofs for two inequalities needed for the proof of Theorem \ref{expexist}. For both of them we use a similar technique of subdividing into many small sub-domains and verifying a slightly stronger but linear inequality on each of them. The proof for inequality~\eqref{eq:pairexp3} demonstrates this technique more clearly so we chose to give it first.

\begin{lemma}\label{comppairexpansion} For $\delta = 0.325$, $\gamma = 1$, and $p_\alpha = 0.45$ the following inequality
\begin{align}
H(x)(1-p_{\alpha})+2\gamma p_\alpha x H\left(\frac{1}{2\gamma}\right) +H(y)> \nonumber \\ \delta H\left(\frac{y}{\delta} \right) + (1-\delta)H\left(\frac{\alpha-y}{1-\delta}\right) + \gamma x H\left(\frac{y}{\gamma x} \right)
\end{align}
holds for any $x \in [0.3,0.3322]$, $y \in [0,\gamma x] \cap [x+\delta-1,\delta]$
\end{lemma}
\begin{proof}
We plug in the convenient value $\gamma = 1$.

Divide the interval $[0.3,0.3322]$ evenly in 1000 sub-intervals $X_1$, $\dots$, $X_{1000}$. For each $X_i = [x^i_{min},x^i_{max}]$ we compute the maximal possible $y$ as $\min(x^i_{max},\delta)$ and divide the interval $$[0, \min(x^i_{max},\delta)]$$ (note that the bound $y \geq x+\delta-1$ is ineffective for considered $x$ and $\delta$) evenly to 1000 sub-intervals $Y_1, \dots, Y_{1000}$. Then for each $Y_j = [y^j_{min},y^j_{max}]$ we compute tight bounds for the expressions
$$V_1 = x, \quad V_2 = y, \quad V_3 = \frac{y}{\delta}, \quad  V_4=\frac{x-y}{1-\delta}, \quad V_5 = \frac{y}{x}$$ that appear as parameters of the function $H$, respectively as
\begin{align*}
I_1 &= [x^i_{min},x^i_{max}], \qquad I_2 = [y^j_{min},y^j_{max}], \qquad &I_3 = \left[\frac{y^j_{min}}{\delta},\frac{y^j_{max}}{\delta}\right], \\
 I_4 &= \left[\frac{x^i_{min}-y^j_{max}}{1-\delta},\frac{x^i_{max}-y^j_{min}}{1-\delta} \right], \qquad &I_5 = \left[\frac{y^j_{min}}{x^i_{max}},\frac{y^j_{max}}{x^i_{min}} \right]
\end{align*}

possibly truncated to $[0,1]$.

For $i=1,2$ we approximate $H(x)$ on $I_i = [p_i,q_i]$ from below with a linear function $H_i$ connecting the points $[p_i,H(p_i)]$ and $[q_i,H(q_i)]$. As $H(x)$ is concave, we indeed have $H(x) \geq H_i(x)$ on $I_i$.

For $i = 3,4,5$ we approximate $H(x)$ on $I_i = [p_i,q_i]$ from above with a linear function $H^i$ that is a tangent to the graph of $H(x)$ at the point $(p_i+q_i)/2$. Due to concavity of $H(x)$, we indeed have $H(x) \leq H^i(x)$ on (not only) $I_i$.

The stronger inequality 
\begin{align}
&H_1(x)(1-p_{\alpha})+2p_\alpha x H\left(\frac12\right) +H_2(y)\nonumber \\
&> \delta H^3\left(\frac{y}{\delta} \right) + (1-\delta)H^4\left(\frac{x-y}{1-\delta}\right) + x H^5\left(\frac{y}{x} \right) \label{eq_comppair2}
\end{align}
is linear in both $x$ and $y$ and thus can be checked only at extreme points of the domain $D \subset X_i \times Y_j$.

These are (in the form $(x,y)$)
\begin{eqnarray*}
\left(\max(x^i_{min},y^j_{min}),y^j_{min}\right),\qquad (x^i_{max},y^j_{min}),\\
\left(\max(x^i_{min},y^j_{max}),y^j_{max}\right),\qquad (x^i_{max},y^j_{max}).
\end{eqnarray*}

Checking these values proves the inequality on $D$ and applying the same procedure for all $i,j \in \{1, \dots, 1000\}$ leads to the full proof. A computer program checking for each of the $4 \cdot 10^6$ values that the left-hand side of (\ref{eq_comppair2}) is greater than the right-hand side by at least 0.0001 has been made available \cite{1}.

\end{proof}

\begin{lemma}For $\delta = 0.325$ and $\Delta = 0.18$ the following inequality

\begin{align}
H(x)(1-\Delta)+\Delta e(x) H\left( \frac{x}{e(x)} \right) + H(y) >  \nonumber\\
\delta H\left(\frac{y}{\delta}\right) + (1-\delta)H\left(\frac{x-y}{1-\delta} \right)+ e(x) H\left( \frac{y}{ e(x)} \right)
\end{align}
holds for any $x \in [0.21,0.48]$, $y \in [0,x] \cap [x+\delta-1,\delta]$,
with $e(x)$ given by the constants
\begin{align*}
C_1 &= 0.2301, \quad &C_3 = 0.3322, \quad C_2 = C_5 = 1-C_1-C_3, \\ C_4 &= 1-C_2, \quad &C_6 = 1-C_3.
\end{align*}
\end{lemma}
\begin{proof} 
We proceed in the same spirit as in the previous lemma. This time we verify the inequality on the intervals $[0.21,C_1]$, $[C_1,C_3]$, $[C_3,C_5]$, and $[C_5,0.48]$, where $e(x)$ is linear, separately. We divide each of the intervals evenly in 1000 sub-intervals $X_1, \dots, X_{1000}$. For each $X_i = [x^i_{min},x^i_{max}]$ we compute the minimum possible $y$ as $$\max(0,x^i_{min}-1+\delta)$$ which equals $0$ for considered $x$ and $\delta$ and maximal possible $y$ as $\min(x_{max},\delta)$ and divide the interval $$[0, \min(x_{max},\delta)]$$ evenly to 1000 sub-intervals $Y_1, \dots, Y_{1000}$. 

Then for each $Y_j = [y^j_{min},y^j_{max}]$ we compute bounds for the expressions
\begin{eqnarray*}
V_1 = x, \quad V_2 = \frac{x}{e(x)}, \quad V_3 = y \\
\quad V_4 = \frac{y}{\delta}, \quad  V_5=\frac{x-y}{1-\delta}, \quad V_6 = \frac{y}{e(x)}.
\end{eqnarray*}
For $i=1,3, 4,5$ these are the same as in Lemma \ref{comppairexpansion} and for $i=2, 6$ we set $I_2$ and $I_6$ respectively as
$$
\left[\min(x^i_{min}/e(x^i_{min}), x^i_{max}/e(x^i_{min})), \max(x^i_{min}/e(x^i_{min}), x^i_{max}/e(x^i_{min})) \right]$$
and
$$
\left[\frac{y^j_{min}}{e(x^i_{max})},\frac{y^j_{max}}{e(x^i_{min})} \right].
$$
All intervals are possibly truncated to $[0,1]$.

For $i=1,2,3$ we approximate $H(x)$ on $I_i$ from below and for $i=4,5,6$ from above as in Lemma \ref{comppairexpansion}.

The stronger inequality
\begin{align}
&H_1(x)(1-\Delta)+\Delta e(x) H_2\left( \frac{x}{e(x)} \right) + H_3(y) \nonumber\\
&< 
\delta H^4\left(\frac{y}{\delta}\right) + (1-\delta)H^5\left(\frac{x-y}{1-\delta} \right)+ e(x) H^6\left( \frac{y}{ e(x)} \right) \label{eq:compexp2}
\end{align}
is linear in $x$ and $y$ (note that we are inside one of the intervals $[0.21,C_1]$, $[C_1,C_3]$, $[C_3,C_5]$, and $[C_5,0.48]$), and thus can be checked only at extreme points of the domain $D \subset X_i \times Y_j$.

These are again (in the form $(x,y)$)
\begin{eqnarray*}
\left(\max(x^i_{min},y^j_{min}),y^j_{min}\right),\qquad (x^i_{max},y^j_{min}),\\
\left(\max(x^i_{min},y^j_{max}),y^j_{max}\right),\qquad (x^i_{max},y^j_{max}).
\end{eqnarray*}

Again a computer program \cite{1} checks for each of the $4 \cdot 10^6$ values that the left-hand side of (\ref{eq:compexp2}) is greater than the right-hand side by at least 0.0001. 

This is done for each of the intervals $[0.21,C_1]$, $[C_1,C_3]$, $[C_3,C_5]$, and $[C_5,0.48]$ which then concludes the proof.

\end{proof}

\end{document}